%% file: DTO.tex
\documentclass[10pt,journal,compsoc]{IEEEtran}[10pt]

\usepackage{cite}
\usepackage{times}
\usepackage{dsfont}
\usepackage{cite}
\usepackage{times}
\usepackage{amsthm}
\usepackage{graphicx}
\usepackage{subfigure}
\usepackage{color}
\usepackage{ifpdf}
\usepackage{CJK}
\usepackage{epsfig}
\usepackage{latexsym}
\usepackage{amsfonts}
\usepackage{amssymb}
\usepackage{paralist}
\usepackage{comment}
\usepackage{xspace}
\usepackage{mathrsfs}
\usepackage{amssymb}
\usepackage{diagbox}
\usepackage{setspace}
\usepackage{color}
\usepackage[small]{caption2}
\usepackage{balance}
\usepackage{algorithm}
\usepackage[noend]{algpseudocode}
\usepackage{xspace}
\usepackage{multirow}

\usepackage{amssymb}
\usepackage{url,epsfig,array}
\usepackage{leftidx}
\usepackage{amsmath}
\usepackage{url}
\usepackage{amsmath}

\newcommand*{\circled}[1]{\lower.7ex\hbox{\tikz\draw (0pt, 0pt)%
    circle (.5em) node {\makebox[1em][c]{\small #1}};}}

\def\ie{\textit{i.e.}\xspace}
\def\etal{\textit{et al.}\xspace}
\def\etc{\textit{etc.}\xspace}

\def\eg{\textit{e.g.}\xspace}

\def\wrt{\textit{w.r.t.}\xspace}

\newtheorem{lemma}{Lemma}

\newcommand{\bluenote}[1]{\textcolor{blue}{#1}}

\begin{document}

\title{Collaborative Inference for Large Models with Task Offloading and Early Exiting}

\author{
        Zuan Xie,~
 	Yang Xu,~\IEEEmembership{Member,~IEEE,}~
 	Hongli Xu,~\IEEEmembership{Member,~IEEE,}~
        Yunming Liao,~ 
        Zhiyuan Yao~
 	\IEEEcompsocitemizethanks{
 		\IEEEcompsocthanksitem Z. Xie, Y. Xu, H. Xu, Y. Liao and Z. Yao are with the School of Computer Science and Technology, University of Science and Technology of China, Hefei, Anhui, China, 230027, and also with Suzhou Institute for Advanced Research, University of Science and Technology of China, Suzhou, Jiangsu, China, 215123. E-mails: xz1314@mail.ustc.edu.cn; xuyangcs@ustc.edu.cn; xuhongli@ustc.edu.cn; ymliao98@mail.ustc.edu.cn; zhiweiyao@mail.ustc.edu.cn.
 	}
}

\IEEEtitleabstractindextext{%
\begin{abstract}
\input{content/abstract.tex}
\end{abstract}
}

\maketitle
\IEEEdisplaynontitleabstractindextext
\IEEEpeerreviewmaketitle


 

\section{Introduction}
\input{content/introduction.tex}

\section{Preliminaries and Collaborative Inference}\label{sec:prelim}

\input{content/prelim.tex}
\section{Algorithm Design}\label{sec:alg}
\input{content/algorithm2.tex}


\section{Experimentation and Evaluation}\label{sec:results}
\input{content/experimentation.tex}

\section{Related Work}\label{sec:related}
\input{content/related.tex}

\vspace{0.1cm}
\section{Conclusions}\label{sec:conclusion}

\input{content/conclusion.tex}




\newpage 
\balance
\bibliographystyle{IEEEtran}
\bibliography{content/refs}

\end{document}

%% file: content/abstract.tex
In 5G smart cities, edge computing is employed to provide nearby computing services for end devices, and the large-scale models (\eg, GPT and LLaMA) can be deployed at the network edge to boost the service quality. 
However, due to the constraints of memory size and computing capacity, it is difficult to run these large-scale models on a single edge node. 
To meet the resource constraints, a large-scale model can be partitioned into multiple sub-models and deployed across multiple edge nodes.
Then tasks are offloaded to the edge nodes for collaborative inference.
Additionally, we incorporate the early exit mechanism to further accelerate inference.
However, the heterogeneous system and dynamic environment will significantly affect the inference efficiency.
To address these challenges, we theoretically analyze the coupled relationship between task offloading strategy and confidence thresholds, and develop a distributed algorithm, termed DTO-EE, based on the coupled relationship and convex optimization.
DTO-EE enables each edge node to jointly optimize its offloading strategy and the confidence threshold, so as to achieve a promising trade-off between response delay and inference accuracy.
The experimental results show that DTO-EE can reduce the average response delay by 21\%-41\% and improve the inference accuracy by 1\%-4\%, compared to the baselines.


%% file: content/introduction.tex
With the development of 5G smart cities, smart street lights, as a basic infrastructure, amounted to 23 million units worldwide by the end of 2022 and are expected to reach 63.8 million by 2027 \cite{report2023}. 
Besides the lighting function, these street lights can serve as a general platform for many other public services. 
For example, surveillance cameras can be deployed to analyze traffic patterns, monitor public safety, or detect abnormal accidents \cite{yong2023implementation,mahoor2020state}.
With the help of 5G edge servers, data from various end devices can be processed and analyzed near the data source, significantly improving the response performance of real-time services\cite{shi2016edge}. 
To further boost the performance of these services, deep learning (DL) models, especially the large-scale models (\eg, GPT\cite{chatgpt} and LLaMA\cite{llama}),  can be deployed at the network edge to enable more intelligent data analysis\cite{gruver2024large}.
However, limited by the memory size and computing capacity of the edge nodes (\eg, surveillance cameras or edge servers), it is difficult even infeasible to run these large-scale models on each edge node.

To meet the resource constraints on edge servers, existing works mainly focus on model compression \cite{ganesh2021compressing,choudhary2020comprehensive} and model partitioning \cite{gao2023Task,fan2023joint,exit1}. 
Although model compression can significantly reduce computational and storage requirements through some advanced techniques such as quantization and pruning, it often results in a certain or even unbearable accuracy loss. 
Additionally, many model compression techniques rely on specific software or hardware support (\eg, low-precision calculations with TensorRT and GPU\cite{li2021unleashing}), or require much extra time for additional adjustments (\eg, architecture search for pruning\cite{dong2019network}).
On the other hand, model partitioning can coordinate multiple nodes to train or infer large-scale models without compromising accuracy.
In edge scenarios, a large-scale model can be divided into several sub-models, each of which is deployed to an edge node. 
Accordingly, a complete inference task is divided into several subtasks, which are processed by the nodes holding the corresponding sub-models with reduced memory and computing burden.
Considering the dependency relationship of these sub-models, a predecessor node (\ie, a node holds a sub-model whose output is the input of another sub-model) will transmit the intermediate features (\ie, the output of an intermediate sub-model) to a successor node, until the inference through all sub-models is completed.
Some works\cite{fan2023joint,exit1} propose to divide the model into two parts and separately distribute to the edge nodes for collaborative inference.
However, this coarse-grained model partitioning exhibits limited effectiveness in reducing the size of sub-models, which may still lead to overload on the nodes.
To further reduce the computing load, some works propose to coordinate multiple edge nodes with a fine-grained model partitioning\cite{li2024distributed,tang2020joint}.
For instance, Tang \textit{et al}.\cite{tang2020joint} divide a large-scale model into multiple (\textgreater2) sub-models and propose a multi-task learning approach to assign each sub-model to an appropriate edge node.

Although fine-grained model partitioning enables offloading the light subtasks to better adapt the capabilities of different edge nodes, the many-to-many dependency between edge nodes (\ie, a node may receive the intermediate features from multiple predecessor nodes and can transmit its features to multiple successors) leads to more complex task offloading strategies.
Therefore, it is technically challenging to offload each subtask to an appropriate edge node that prevents node overload and enhances resource utilization.
In addition, the optimization of task offloading also suffers from two other critical challenges in practical applications.
\textbf{1) Heterogeneous System.} Edge nodes typically possess limited and heterogeneous resources\cite{li2022pyramidfl,lim2020federated,liao2024mergesfl,liao2023accelerating}. 
For instance, computing and bandwidth capacities may vary by more than tenfold among different edge nodes\cite{chen2022decentralized,lai2021oort,liao2024asynchronous}. 
Given the same computing load across different edge nodes, the system heterogeneity will lead to long response delay on weak nodes and resource waste on strong nodes, significantly impacting the response efficiency.
\textbf{2) Dynamic Environment.} The task arrival rate changes over time and space\cite{ai2023smart,liao2023adaptive,liao2024parallelsfl,liao2023decentralized}. 
For example, the cameras deployed at a crowded train station will generate more tasks than those at an empty
campus. 
Additionally, the internal conditions of edge nodes, such as temperature and voltage fluctuations, will affect their computing capabilities\cite{patsias2023task}. 
The dynamic environment leads to imbalanced loads for the edge nodes, and the task offloading strategy needs to be adjusted to adapt to the dynamic environment.

Some recent works \cite{tiwary2018response,2partition2} have made the effort to address these challenges.
For instance, Tiwary \textit{et al}. propose a non-cooperative game based task offloading (NGTO) method\cite{tiwary2018response} to minimize the task response delay by assigning appropriate computing loads to heterogeneous edge nodes.
However, NGTO requires each node to gather task offloading strategies of others, and update the task offloading strategy in a round-robin manner, which leads to a long decision time.
Besides, Mohammed \textit{et al}.\cite{2partition2} propose DINA-O, an approach that offloads multiple interdependent subtasks by matching appropriate edge nodes to balance the load and improve resource utilization.
But it assumes fixed computing resources on edge nodes, and requires each node can obtain the status information (\eg, available computing resources) of all nodes in time, which is unsuitable for the dynamic network environment.

To realize efficient inference for large-scale models at the edge, we propose a collaborative inference framework.
As illustrated in Fig. \ref{fig:edgecol}, a large-scale model is divided into multiple sub-models, each of which is deployed across multiple edge servers.
In order to accelerate task inference, the early exit mechanism is employed in our framework, where some sub-models are mounted with exit branches.
Generally, the tasks are continuously produced by end devices and delivered to edge servers for collaborative inference.
The edge servers are enabled to develop the task offloading strategies in a distributed fashion to balance the computing loads of their heterogeneous successor servers.
The distributed manner enables fast convergence of offloading strategy on each server, enhancing adaptability to the dynamic environment.

\begin{figure}[!t]
\centering
\includegraphics[width=0.95\linewidth]{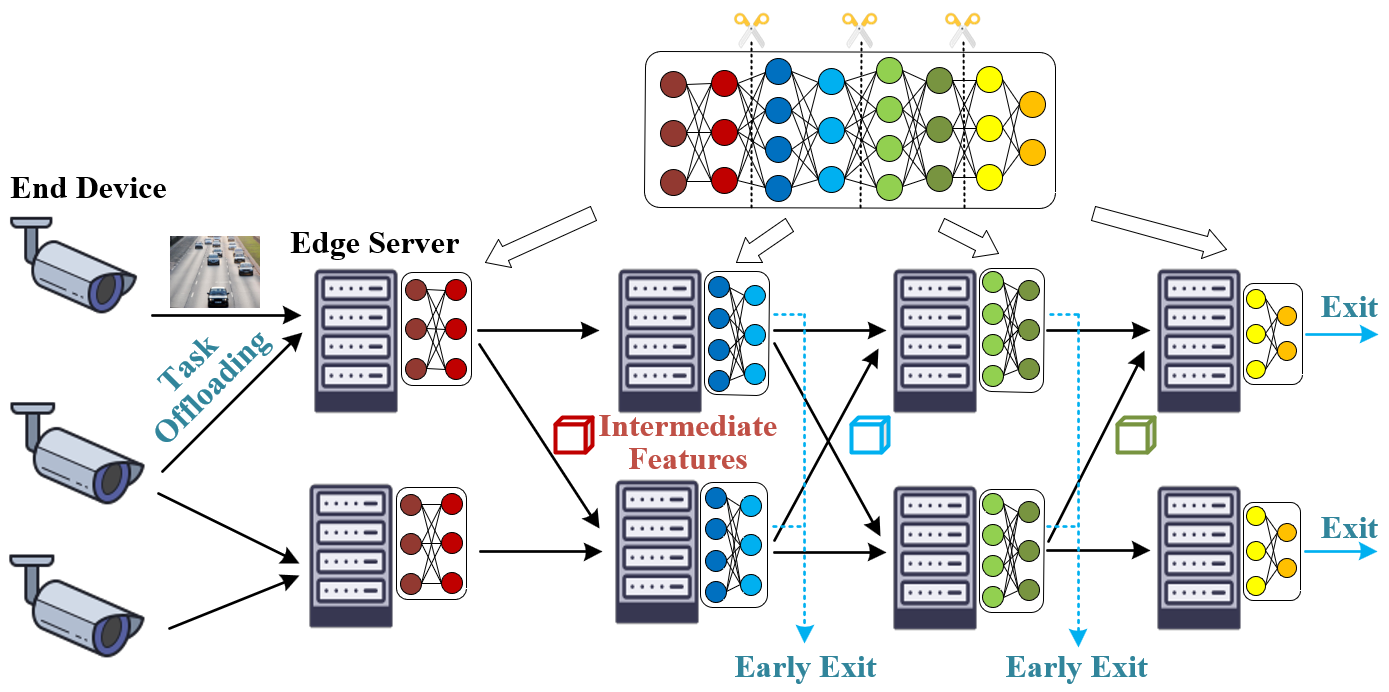}
\caption{\label{fig:edgecol}Illustration of collaborative inference of a large model with multiple edge nodes.}
\vspace{-4mm}
\end{figure}

Though an efficient task offloading strategy can enhance resource utilization and reduce response delay, the incorporation of early exit mechanism may complicate the development of offloading strategy.
Usually, a lower early exit confidence threshold allows more tasks to exit at the earlier edge servers\cite{laskaridis2021adaptive}, which contributes to shorter inference delay and less computing load on the subsequent servers, but potentially results in more accuracy loss, and vice versa.
Furthermore, the task offloading strategy also influences how computing loads are distributed across different edge servers, which in turn affects subsequent decisions regarding both the offloading strategy and the confidence threshold.
Therefore, we need to consider the interaction between the task offloading strategy and the confidence threshold, and jointly optimize them to achieve the ideal trade-off between task response delay and inference accuracy.
Our main contributions are as follows:
\begin{itemize}
\item 
To enable the deployment of large-scale models at the edge, we propose a collaborative inference framework, which coordinates multiple edge nodes and employs early exit mechanism to accelerate task inference.
\item 
We theoretically analyze the coupled relationship between task offloading strategy and confidence thresholds of different exit branches.
Based on the coupled relationship and convex optimization, we develop a distributed algorithm, called DTO-EE, which enables each edge node to jointly optimize its offloading strategy and the confidence threshold, so as to achieve a promising trade-off between response delay and inference accuracy.
\item 
We conduct extensive simulation experiments to evaluate the performance of DTO-EE under various system settings. 
The experimental results demonstrate that DTO-EE can reduce the average response delay by 21\%-41\% and improve the inference accuracy by 1\%-4\%, compared to the baselines.
\end{itemize}

%% file: content/prelim.tex
\subsection{Model Partitioning}
The collaborative inference is based on model partitioning.
Generally, the complete model can be regarded as a function $M(\cdot)$ that maps the input $x$ to the prediction $M(x)$. 
Assume that the complete model is divided into $H$ sub-models, denoted as $M_{1},M_{2},...,M_{H}$, the output of each sub-model serves as the input of its next sub-model.
Then, the model inference can be expressed as:
\begin{equation}
	\hat{y} = M(x) = (M_{H} \circ M_{H-1} \circ \cdots \circ M_{1})(x)
\end{equation}
where $x$ is the input and $\hat{y}$ is the prediction from the model. 
The operator $\circ$ denotes the combination, \ie, $(M_{h} \circ M_{h-1})(\cdot) = M_{h}(M_{h-1}(\cdot))$.
The output of sub-model $M_{h}$ is denoted as $z_{h}$ ($h \in [1, H]$).
We denote $z_{0} = x$ and $z_{h} = M_{h}(z_{h-1})$.

In addition, we add early exit branches to certain sub-models.
Then, we feed the output of the sub-model $M_{h}$ to the corresponding branch (\ie, $b_{h}$) to obtain an early prediction:
\begin{equation}
	\hat{y}_{h} = b_{h}(z_{h}), \text{if}\ {E}_{h}=1
\end{equation}
where ${E}_{h}=1$ means that the sub-model $M_{h}$ has an early exit branch $b_{h}$, and $z_{h}$ is the input of the branch $b_{h}$.


\subsection{Network Model}
We consider an edge network consisting of $\mathcal{U}$ edge servers (ESs) and $\mathcal{V}$ end devices (EDs), collectively referred to as edge nodes.
We use $e^{h}_{i}$ ($i \in [1,n_{h}]$) to represent the ES with a sub-model $M_{h}$. 
Each ES holds a sub-model, and we use $n_{h}$ to represent the number of ESs that hold the sub-model $M_{h}$.
In particular, we use $e^{0}_{i}$ ($i \in [1,\mathcal{V}]$) to represent an ED.
The set of ESs possessing the same sub-model $M_{h}$ is denoted as $S^{h}$, where $|S^{h}|=n_{h}$, and the set of all ESs is denoted as $S=S^{1}\cup S^{2}\cdots \cup S^{H} = \{e^{1}_{1},\cdots,e^{1}_{n_{1}},e^{2}_{1},\cdots,e^{2}_{n_{2}}, \cdots,e^{H}_{1},\cdots,e^{H}_{n_{H}}\}$, where $|S|=\mathcal{U}$. 
Each ES $e^{h}_{i} \in S$ has a maximum computing capacity, denoted as $\mu^{h}_{i}$, which is measured by the number of floating-point operations per unit time.
The set of all EDs is denoted as $D =\{e^{0}_{1},e^{0}_{2},\cdots,e^{0}_{\mathcal{V}}\}$, where $|D|=\mathcal{V}$.
Due to the directional nature of task offloading and the dependency between sub-models, the edge nodes are arranged to work in a pipeline manner as illustrated in Fig. \ref{fig:edgecol}. 
Thus, task offloading happens between directly connected nodes.
We denote the set of successors of node $e^{h}_{i}$ as $L^{h}_{i}$, which will receive subtasks from node $e^{h}_{i}$, and the set of predecessors of node $e^{h}_{i}$ is denoted as $V^{h}_{i}$, which may offload subtasks to node $e^{h}_{i}$.
The transmission rate between nodes $e^{h}_{i}$ and $e^{h+1}_{j}$ is denoted as $r^{h}_{i,j}$.
The offloading strategy of node $e^{h}_{i}$ is denoted as the vector $p^{h}_{i} = \{p^{h}_{i,j}|e_{j}\in L^{h}_{i}\}$, where $p^{h}_{i,j}$ is the probability that $e^{h}_{i}$ offloads its tasks to $e^{h+1}_{j}$.
We denote $P=\{p^{h}_{i}|e^{h}_{i}\in S\cup D\}$ as the task offloading strategies of all nodes.
For clarity, we summarize the important notations in Table 1. 

\begin{table}[!t]
\setlength{\abovecaptionskip}{10pt}%
\setlength{\belowcaptionskip}{0pt}%
\centering
\caption{Key Notations.}
\label{tbl:natation}
\resizebox{\linewidth}{!}{
\begin{tabular}{cl}
\hline
\textbf{Symbol} & \textbf{Semantics}\\
\hline
$D$ & the set of EDs $\{e^{0}_{1},e^{0}_{2},\cdots,e^{0}_{V}\}$ 
\\
$S$ & the set of ESs $\{e^{1}_{1},\cdots,e^{1}_{n_{1}},e^{2}_{1},\cdots,e^{2}_{n_{2}}, \cdots,e^{H}_{1},\cdots,e^{H}_{n_{H}}\}$
\\ 
$S^{h}$ & the set of ESs holding the sub-model $M_{h}$
\\
$V^{h}_{i}$ & the set of predecessor edge nodes of $e^{h}_{i}$
\\
$L^{h}_{i}$ & the set of successor edge nodes of $e^{h}_{i}$
\\
$\phi^{h}_{i}$ & task arrival rate of edge node $e^{h}_{i}$ 
\\
$\Phi$ & total task arrival rate of the system 
\\
$\mu^{h}_{j}$ & the maximum computing capacity of edge node $e^{h}_{j}$
\\
$\alpha_{h}$ & the computing resources required for inference on sub-model $M_{h}$ 
\\
$\beta_{h}$ & the size of input data of sub-model $M_{h}$
\\
$p^{h}_{i}$ & the task offloading strategy of edge node $e^{h}_{i}$ 
\\
$r_{i,j}^{h}$ & transmission rate from edge node $e^{h}_{i}$ to $e^{h+1}_{j}$ 
\\
$\lambda^{h}_{j}$ & the total computing resources required by edge node $e^{h}_{j}$ 
\\
\hline
\end{tabular}
}
\vspace{-4mm}
\end{table}

\subsection{Collaborative Inference}
Tasks arrive at each ED $e^{0}_{i} \in D$ following a Poisson process with an expected arrival rate of $\phi^{0}_{i}$.
The total task arrival rate of the system is denoted as $\Phi = \sum_{e^{0}_{i} \in D} \phi^{0}_{i}$.
Upon the arrival of each task at an ED, it will be offloaded to a selected ES according to the ED's task offloading strategy. 
After the ES completes the inference with its sub-model, it will offload the intermediate features to a successor for further inference.
Additionally, for each ES that holds a sub-model with an early exit branch, the task may exit early if its predicted probability exceeds the confidence threshold.
We use $C=\{c_{h}|E_{h}=1,h\in[1,H]\}$ to represent the setting of the confidence thresholds, where $c_{h}$ is a discrete value specifying the threshold for the exit branch of sub-model $M_{h}$.
The confidence thresholds influence the proportion of tasks exiting early and the overall inference accuracy.
During the task inference process, some simple tasks will exit early at a current sub-model, while the remaining will be offloaded to the next sub-models for further analysis.
We use $I_{h}$ to denote the estimated proportion of remaining tasks in sub-model $M_{h}$, with $I_{h}=1$ indicating no exit branch in sub-model $M_{h}$.
Besides, branches in different sub-models will lead to varying accuracies for task predictions. 
Thus,  given the confidence thresholds in $C$, the estimated overall accuracy is denoted as $\mathcal{A}(C)$.
The highest and lowest achievable inference accuracies across all threshold settings are denoted as $\mathcal{A}_{max}$ and $\mathcal{A}_{min}$, respectively. 
Generally, $\mathcal{A}_{max}$ and $\mathcal{A}_{min}$ correspond to scenarios where all tasks propagate through all sub-models or exit at the earliest branch, respectively.
The task arrival rate of each ES $e^{h}_{j}$, denoted as $\phi^{h}_{j}$, is influenced by both the confidence threshold and the task offloading strategies of its predecessors in $V^{h}_{j}$.
Under the steady condition where the task input and output rates are balanced, the arrival rate of tasks offloading from node $e^{h-1}_{i}$ (\eg, an ED or an ES) to node $e^{h}_{j}$ (\ie, an ES) is $\varphi^{h-1}_{i,j}=p^{h-1}_{i,j} \phi^{h-1}_{i} I_{h-1}$.
Then $\phi^{h}_{j}$ ($e^{h}_{j}\in S$) is given by:
\vspace{-0.1cm}
\begin{equation}
	\phi^{h}_{j} = \sum\limits_{e^{h-1}_{i} \in V^{h}_{j}} \varphi^{h-1}_{i,j} = \sum\limits_{e^{h-1}_{i} \in V^{h}_{j}} p^{h-1}_{i,j} \phi^{h-1}_{i} I_{h-1}
	\label{newphi}
 \vspace{-0.2cm}
\end{equation}
The response delay of a task involves multiple offloading delays of consecutive subtasks, each of which includes transmission delay and computation delay.
Similar to many studies \cite{zhang2016energy, chen2015efficient}, we neglect the time for returning the inference results to EDs, as the size of results is usually small.
The number of floating point operations required for a subtask on sub-model $M_{h}$ is denoted as $\alpha_{h}$, and the input data size of sub-model $M_{h}$ is denoted as $\beta_{h}$.
The transmission delay of a subtask from node $e^{h}_{i}$ ($i\in [1,n_{h}]$) to $e^{h+1}_{j}$ ($j\in [1,n_{h+1}]$) is calculated as:
\vspace{-0.1cm}
\begin{equation}
	T^{cm,h}_{i,j} = \beta_{h+1}/r^{h}_{i,j}
 \vspace{-0.2cm}
\end{equation}


Given the offloading decision $p^{h}_{i,j}$, the required computing resources on node $e^{h+1}_{j}$ for the subtasks offloaded from node $e^{h}_{i}$ is $\xi^{h}_{i,j} = p^{h}_{i,j}\phi^{h}_{i}I_{h}\alpha_{h+1}$.
Accordingly, the required computing resources on $e^{h+1}_{j}$ is expressed as:
\vspace{-0.1cm}
\begin{equation}
	\lambda^{h+1}_{j} = \sum_{e^{h}_{i} \in V^{h+1}_{j}}\xi^{h}_{i,j} = \sum_{e^{h}_{i} \in V^{h+1}_{j}}p^{h}_{i,j}\phi^{h}_{i}I_{h}\alpha_{h+1}
    \label{newlambda}
    \vspace{-0.2cm}
\end{equation}

On basis of the splitting property of Poisson process\cite{altman2011load}, since tasks' arrival on each ED follows a Poisson process, the tasks' arrival on each ES within the set $S^{1}$ also follows the Poisson process.
Then the computing delay at each ES within the set $S^{1}$ can be modeled as an M/D/1-PS queuing system.
Furthermore, according to the properties of the M/D/1-PS queue, the output process remains a Poisson process, which is similar to the input process\cite{haviv2013course}. 
Consequently, the input processes of subsequent ESs still follow the Poisson process, and each of them can be modeled as an M/D/1-PS queuing system.
Therefore, the computing delay for a subtask on each ES $e^{h}_{i}$ is given by:
\vspace{-0.3cm}
\begin{equation}
T^{cp,h}_{i} = \frac{\alpha_{h}}{\mu^{h}_{i} - \lambda^{h}_{j}}	
\vspace{-0.2cm}
\end{equation}

\vspace{-0.1cm}
\subsection{Problem Definition}
As suggested in \cite{ma2023fully}, we divide the system timeline into fixed-length time slots (\eg, 5s), each comprising a configuration update phase and a task offloading phase.
During the configuration update phase, we need to jointly optimize the task offloading strategy and the confidence threshold, which contributes to achieving the trade-off between average response delay and inference accuracy in the task offloading phase. 


The average response delay of tasks involves the average offloading delays of $H$ subtasks.
We employ $T^{h}_{j}$ to denote the average offloading delay of subtasks offloaded from all nodes in $V^{h}_{i}$ to $e^{h}_{j}\in S$, \ie,
\vspace{-0.1cm}
\begin{equation}
T^{h}_{j} = T^{cp,h}_{j} + \frac{1}{\phi^{h}_{j}}\sum_{e^{h-1}_{i}\in V^{h}_{j}} \varphi^{h-1}_{i,j} T^{cm,h-1}_{i,j}
\vspace{-0.1cm}
\end{equation}
The average response delay $T$ of all tasks in the system is:
\vspace{-0.1cm}
\begin{equation}
T = \frac{1}{\Phi} \sum_{e^{h}_{j}\in S}(\frac{\lambda^{h}_{j}}{\mu^{h}_{j} - \lambda^{h}_{j}} +\!\!\!\! \sum_{e^{h-1}_{i}\in V^{h}_{j}}\!\!\!\! \varphi^{h-1}_{i,j} T^{cm,h-1}_{i,j})
\vspace{-0.1cm}
\end{equation}



To trade off the average response delay $T$ and inference accuracy $\mathcal{A}$, we introduce a utility function $U(T,\mathcal{A})$ to evaluate system performance as:
\begin{equation}
	U(T,\mathcal{A}) = a T - (1-a) \frac{(\mathcal{A} - \mathcal{A}_{min})}{(\mathcal{A}_{max} - \mathcal{A}_{min})}
\end{equation}
where $a$ is a weight coefficient used to balance the performance of inference delay and accuracy, and the accuracy $\mathcal{A}$ is the normalized representation.
We need to minimize the system utility function $U$ while satisfying the resource constraint.
The problem can be formalized as follows:
\begin{equation}
\textbf{(P1)}:\min\limits_{P,C} U(T,\mathcal{A}(C))\nonumber
\label{funU}
\end{equation}
\begin{equation}
s.t.\left\{
\begin{aligned}
&\lambda^{h}_{j} < \mu^{h}_{j},  &\forall  e^{h}_{j} \in S\\
&\sum\nolimits_{e^{h+1}_{j} \in L^{h}_{i}} p^{h}_{i,j} = 1, &\forall  e^{h}_{i} \in S\cup D\\
&0 \leq p^{h}_{i,j} \leq 1, &\forall  e^{h}_{i}, e^{h+1}_{j} \in S\cup D
\end{aligned}
\right.
\end{equation}
The first set of inequalities ensures that the total computing requirement on each ES should not exceed its maximum computing capacity. 
The last two conditions ensure that each task will be successfully offloaded to only one ES.

%% file: content/algorithm2.tex
\begin{figure}[!t]
	\centering
	\includegraphics[width=0.85\linewidth]{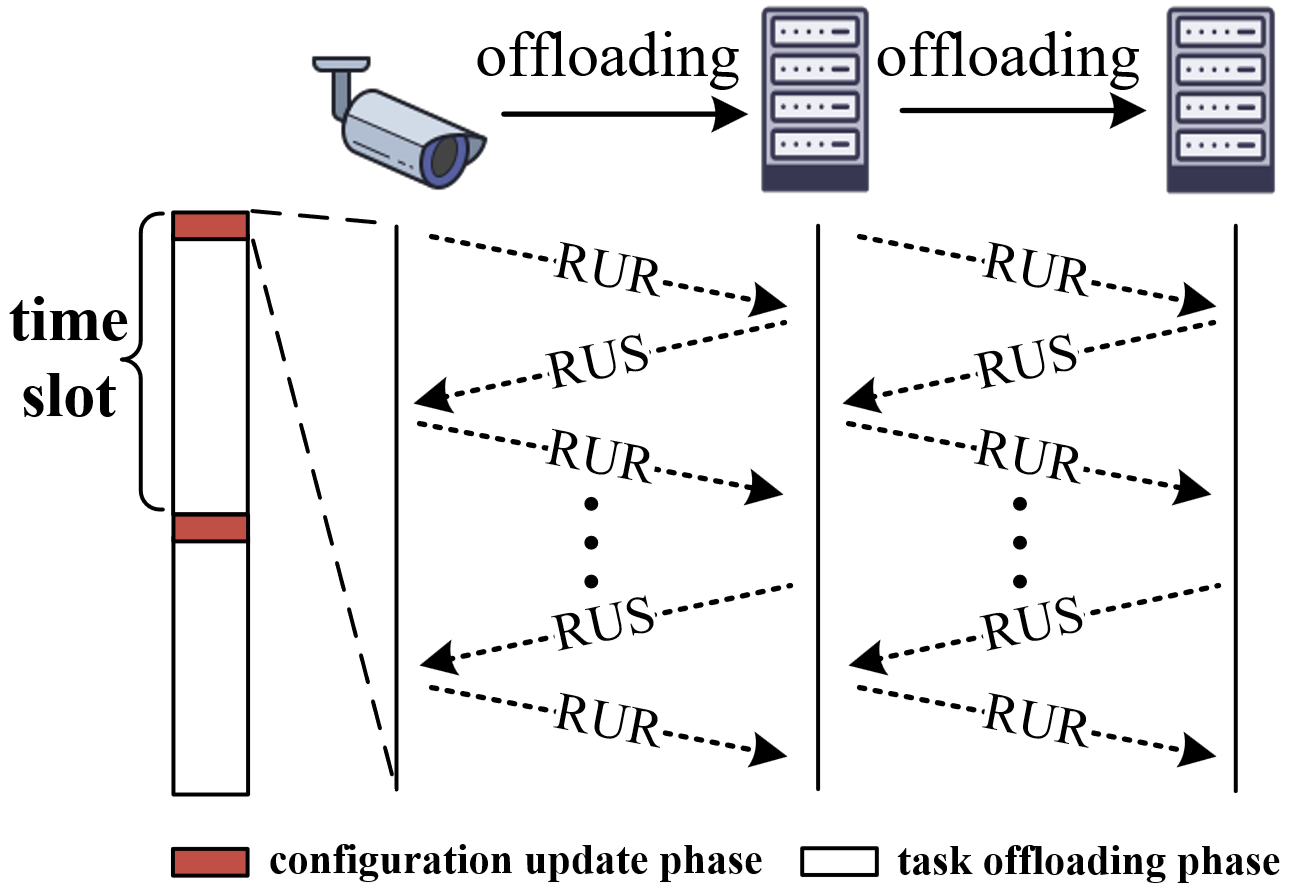}
	\caption{\label{fig:sysprocess} The overview of DTO-EE.}
    \vspace{-5mm}
\end{figure}

We first present an overview of the DTO-EE (Section \ref{section:DTO-EE}), as shown in Fig. \ref{fig:sysprocess}.
During the configuration update phase, the DTO-EE algorithm joint optimizes the task offloading strategy and confidence threshold.
Each edge node, as a task offloader or a task receiver, conducts multiple rounds of local communication with its directly connected nodes.
In each communication round, the DTO-R (Section \ref{section:DTO-R}) and the DTO-O (Section \ref{section:DTO-O}) algorithms guide each receiver and offloader to update the task offloading strategy, respectively.
Specifically, the offloaders send Resource-Utilization-Request (RUR) messages to their respective receivers, who then respond with Resource-Utilization-Status (RUS) messages\cite{ma2023fully}.
Based on the RUR messages from the receivers, each offloader updates its task offloading strategy for the next round.
Additionally, the nodes adjust the confidence threshold of an exit branch at the interval of several communication rounds.

\subsection{Problem Transformation and Optimization Analysis}\label{Problem Transformation}
During the configuration update phase, each node first updates the task offloading strategy through several rounds of local communication without changing confidence thresholds. 
We first consider the problem of optimizing the task offloading strategy with fixed confidence thresholds.
To deal with the problem in \textbf{P1} whose solution may violate the computing resource constraint, we first transform the constraint of \textbf{P1} into its optimization objective by the exterior-point method. 
Specifically, we define the penalty function $N(P)$ as follows:
\vspace{-0.1cm}
\begin{equation}
N(P) = K \cdot \sum\nolimits_{e^{h}_{j}\in S}(\max\{0,\lambda^{h}_{j}-\mu^{h}_{j} + \epsilon \})^{2}
\vspace{-0.1cm}
\end{equation}
where $\epsilon$ is an arbitrarily small positive constant, and $K$ is a sufficiently large positive constant, called the penalty factor. 
By adding this penalty function to the objective function, we define a new objective function $R(P) = T+N(P)$. 
With fixed confidence thresholds, the original problem is transformed into a new optimization problem:
\vspace{-0.2cm}
\begin{equation}
\textbf{(P2)}:\min\limits_{P} R(P)\nonumber
\vspace{-0.2cm}
\end{equation}
\vspace{-0.1cm}
\begin{equation}
s.t.\left\{
\begin{aligned}
&0 \leq p^{h}_{i,j} \leq 1,      &\forall e^{h}_{i}\in S\cup D\\
&\sum\nolimits_{e^{h+1}_{j} \in L^{h}_{i}} p^{h}_{i,j} = 1, &\forall e^{h}_{i}\in S\cup D
\end{aligned}
\right.
\vspace{-0.1cm}
\end{equation}
If the offloading strategy $P$ does not satisfy $\lambda^{h}_{j}<\mu^{h}_{j} - \epsilon$ for any node $e^{h}_{j}\in S$ (h\textgreater0), the value of $R(P)$ will be very large. 
Therefore, by setting a sufficiently large value for $K$, we can approximate the optimal solution for \textbf{P1} by solving \textbf{P2}.
The main challenge to obtain the optimal solution of \textbf{P2} is the incompleteness of information in distributed scenarios. 
The offloading strategy made by a single node without the strategies of other nodes may violate the computing resource constraint.

In addition to updating the offloading strategy, the confidence threshold is adjusted at intervals of several rounds of communication to balance the response delay and inference accuracy.
To estimate the impact of the confidence threshold on delay, we first need to analyze the coupled relationship between the confidence threshold and the task offloading strategy.
The early exit of a task can be regarded as offloading the task to a virtual node, which will not conduct task offloading further.
Therefore, changing the confidence threshold for the exit branch on node $e^{h}_{i}$ is equivalent to changing the probability of offloading tasks to the virtual node, which simultaneously scales the probability of offloading tasks to each node $e^{h+1}_{j}\in L^{h}_{i}$.
We denote the new remaining ratio of the adjusted threshold as $I'_{h}$, then the probability of offloading tasks to each node $e^{h+1}_{j}\in L^{h}_{i}$ changes from $p^{h}_{i,j}$ to $\frac{I'_{h}}{I_{h}}p^{h}_{i,j}$.
The partial derivative of the objective function $R(P)$ \wrt $p^{h}_{i,j}$ is:
\vspace{-0.2cm}
\begin{equation}
    \begin{aligned}
    \label{partial_derivative}
        \hspace{-1mm}
        \frac{\partial R(P)}{\partial p_{i,j}^{h}} &= \frac{\phi^{h}_{i} I_{h}}{\Phi}(\frac{\mu^{h+1}_{j}\alpha_{h+1}}{(\mu^{h+1}_{j}-\lambda^{h+1}_{j})^2 }+\frac{\beta_{h+1}}{r^{h}_{i,j}}+\Omega^{h+1}_{j}\\
                       + & 2K\Phi  \cdot \max\{0, \alpha_{h+1}(\lambda^{h+1}_{j}-\mu^{h+1}_{j}+\epsilon)\})
    \end{aligned}
\vspace{-0.1cm}
\end{equation}
where $\Omega^{h+1}_{j}$ is given by:
\vspace{-0.4cm}
\begin{equation}
    \begin{aligned}
    \label{Omega}
        \hspace{-1mm}
        \Omega^{h+1}_{j} &= \!\!\!\!\sum_{e^{h+2}_{k}\in L^{h+1}_{j}}\!\!\!\! p^{h+1}_{j,k} I_{h+1} (\frac{\mu^{h+2}_{k}\alpha_{h+2}}{(\mu^{h+2}_{k}-\lambda^{h+2}_{k})^2 }+\frac{\beta_{h+2}}{r^{h+1}_{j,k}}\\
                + &\Omega^{h+2}_{k} + 2K\Phi \cdot \max\{0, \alpha_{h+2}(\lambda^{h+2}_{k}-\mu^{h+2}_{k}+\epsilon)\})
    \end{aligned}
    \vspace{-0.1cm}
\end{equation}
We define 
\vspace{-0.3cm}
\begin{equation}
    \begin{aligned}
        \hspace{-1mm}
        \Delta^{h}_{i,j} \triangleq &\frac{\mu^{h+1}_{j}\alpha_{h+1}}{(\mu^{h+1}_{j}-\lambda^{h+1}_{j})^2 }+\frac{\beta_{h+1}}{r^{h}_{i,j}}+\Omega^{h+1}_{j}\\
                       + & 2K\Phi \cdot \max\{0, \alpha_{h+1}(\lambda^{h+1}_{j}-\mu^{h+1}_{j}+\epsilon)\}
        \label{newDelta}
    \end{aligned}
\end{equation}
then, 
\begin{equation}
   \begin{aligned}
       \Omega^{h}_{i} = \sum\nolimits_{e^{h+1}_{j}\in L^{h}_{i}} p^{h}_{i,j} I_{h} \Delta_{i,j}^{h}
        \label{newOmega}
   \end{aligned}
\end{equation}

Therefore, the response delay change caused by changing the confidence threshold on node $e^{h}_{i}$ to a certain extent is:
\begin{equation}
    \begin{aligned}
        \hspace{-1mm}
        \Delta D^{h}_{i} =\!\!\!\!\!\! \sum_{e^{h+1}_{j}\in L^{h}_{i}}\!\!\!\!\!\! \frac{\phi^{h}_{i} I_{h}}{\Phi}\Delta^{h}_{i,j} (\frac{I'_{h}-I_{h}}{I_{h}})p^{h}_{i,j} =  \frac{\phi^{h}_{i}}{\Phi}(\frac{I'_{h}-I_{h}}{I_{h}})\Omega^{h}_{i}
        \label{DeltaD}
    \end{aligned}
\end{equation}
Given the change of inference accuracy $\Delta \mathcal{A}$, the impact of changing the confidence threshold of all nodes in $S^{h}$ on the utility function $U$ is given by:
\begin{equation}
    \begin{aligned}
        \hspace{-1mm}
        \Delta U = a \sum\nolimits_{e^{h}_{i} \in S^{h}} \Delta D^{h}_{i} - (1-a) \Delta \mathcal{A}
        \label{DeltaU}
    \end{aligned}
\end{equation}

Based on the changes in inference accuracy $\mathcal{A}$ and remaining rate $I_{h}$, we can estimate the impact of adjusting the confidence threshold on the utility function.
The inference accuracy and remaining rate can be estimated by a validation dataset.
Specifically, we adopt a reuse-based method that can predetermine the outcomes for all confidence thresholds by a one-shot evaluation. 
Given a sample in the validation dataset, its softmax output from each branch is fixed.
The confidence threshold can only decide whether the sample exits early or not from the branch.
Therefore, we record the softmax outputs of all samples at each branch, and these outputs can be reused to screen out the early-exit samples given different settings of confidence thresholds.
In this way, we only need to evaluate the model once to obtain an accuracy-ratio table, which records the inference accuracy and remaining ratio for each setting of confidence thresholds.
After obtaining the new confidence thresholds, each node $e^{h}_{i}\in S$ can update its remaining ratio $I_{h}$ immediately based on the accuracy-ratio table.

\begin{algorithm}[!t]   
	\caption{Task Receiver Algorithm (DTO-R) for $e^{h}_{j}$} \label{DTO-R} 
	\begin{algorithmic}[1]
	  \State Receive RUR from $V^{h}_{j}$ to get $\xi_{\cdot,j}^{h-1,t}$,$C$
        \State Synchronize $C$
        \State Calculate $\lambda_{j}^{h,t}$ by Eq. (\ref{newlambda})
        \State Calculate $\phi^{h,t}_{j} = \lambda^{h,t}_{j}/\alpha_{h}$
        \State Get $\Omega^{h,t}_{j}$\! from Algorithm \ref{DTO-O}
        \State Broadcast RUS including $\lambda_{j}^{h,t}, \Omega^{h,t}_{j}, \mu^{h}_{j}$, $C$ to $V^{h}_{j}$
	\end{algorithmic} 
\end{algorithm}

\subsection{Task Receiver Algorithm (DTO-R)}\label{section:DTO-R}
This section describes the task receiver algorithm DTO-R (Algorithm \ref{DTO-R}), which is concurrently executed by edge servers in each communication round during the configuration update phase.
For ease of description, the subscript $t$ is introduced to indicate the $t$-th communication round.
In communication round $t$, a receiver $e^{h}_{j}\in S$ collects RUR messages from its offloaders in $V^{h}_{j}$ (Algorithm \ref{DTO-R}, Line 1).
An RUR message from the offloader $e^{h-1}_{i}$ includes the confidence thresholds $C$ and the required computing resources on receiver $e^{h}_{j}$ for $e^{h-1}_{i}$, denoted as $\xi^{h-1,t}_{i,j}$. 
The thresholds $C$ and the vector $\xi^{h-1,t}_{\cdot,j} = \{\xi^{h-1,t}_{i,j} | e^{h-1}_{i} \in V^{h}_{j}\}$ constitute the received RURs on receiver $e^{h}_{j}$.
Then, the receiver $e^{h}_{j}$ synchronizes the confidence thresholds $C$ and estimates the total required computing resources $\lambda^{h,t}_{j}$ (Algorithm \ref{DTO-R}, Lines 2-3), and finally responds with RUS messages (Algorithm \ref{DTO-R}, Line 6).
An RUS message includes $\lambda^{h,t}_{j}$, the maximum computing capacity $\mu^{h,t}_{j}$, the confidence thresholds $C$ and the gradient information $\Omega^{h,t}_{j}$.
Note that the gradient information $\Omega^{h,t}_{j}$ reflects the degree of delay fluctuation on the basis of the current load, which is used to guide the update direction of the task offloading strategy.
$\Omega^{h,t}_{j}$ is obtained when the node $e^{h}_{j}$ runs the Algorithm \ref{DTO-O} (Section \ref{section:DTO-O}), and $\Omega^{H,t}_{j}=0$ for each node $e^{H}_{j} \in S^{H}$.

\subsection{Task Offloader Algorithm (DTO-O)}\label{section:DTO-O}

This section describes the task offloader algorithm DTO-O (Algorithm \ref{DTO-O}), which is concurrently executed by all edge nodes not in $S^{H}$ in each communication round during the configuration update phase.

In the $t$-th communication round, the offloader $e^{h}_{i}$ receives the RUS message from each receiver $e^{h+1}_{j} \in L^{h}_{i}$, which includes $\lambda_{j}^{h+1,t}$, $\Omega^{h+1,t}_{j}$, $C$ and $\mu^{h+1}_{j}$ (Algorithm \ref{DTO-O}, Line 1).
Based on the received information, $e^{h}_{i}$ synchronizes the confidence thresholds $C$ and prepares to update its task offloading strategy.
To reduce the response delay, the offloader tends to increase the probability of offloading subtasks to the most idle receiver.
We use $\Delta_{i,j}^{h,t}$ as the repulsive factor of offloader $e^{h}_{i}$ to receiver $e^{h+1}_{j}$ in the $t$-th round, which reflects the effect of the offloading probability $p^{h,t}_{i,j}$ on the response delay.
After calculating the repulsive factor $\Delta_{i,j}^{h,t}$ for each receiver $e^{h+1}_{j} \in L^{h}_{i}$ by Eq. \eqref{newDelta}, offloader $e^{h}_{i}$ obtains its gradient information $\Omega^{h,t}_{i}$ by Eq. \eqref{newOmega}, which reflects the computing load of subsequent nodes (Algorithm \ref{DTO-O}, Lines 3-4). 


\begin{algorithm}[!t] 
	\caption{Task Offloader Algorithm (DTO-O) for $e^{h}_{i}$} 
	\label{DTO-O} 
	\begin{algorithmic}[1]
		\State Receive RUS from $L^{h}_{i}$ to get $\lambda_{j}^{h+1,t}$, $\Omega^{h+1,t}_{j}$, $C$, $\mu^{h+1}_{j}$ of each  $e^{h+1}_{j}$
        \State Synchronize $C$
		\State Calculate $\Delta_{i,j}^{h,t}$ for each $e^{h+1}_{j} \in L^{h}_{i}$ by Eq. (\ref{newDelta})
        \State Calculate $\Omega^{h,t}_{i}$ by Eq. (\ref{newOmega})
		\State Get new offloading strategy $p^{h,t+1}_{i,\cdot}$ by Eq. (\ref{newp})
        \State Get $\phi^{h,t}_{i}$ from Algorithm \ref{DTO-R}
		\For{Each $e^{h+1}_{j} \in L^{h}_{i}$}
		\State Calculate $\xi_{i,j}^{h,t+1} = p_{i,j}^{h,t+1} \phi^{h,t}_{i} I_{h} \alpha_{h+1}$
		\State Send RUR including $\xi_{i,j}^{h,t+1}$, $C$ to $e^{h+1}_{j}$
		\EndFor
	\end{algorithmic} 
\end{algorithm}
The offloaders tend to assign higher offloading probabilities to the receivers with lower repulsive factors. 
Note that as long as the receiver $e^{h+1}_{j}$ is overloaded, $\Delta_{i,j}^{h,t}$ will be very large and $e^{h}_{i}$ will reduce the probability of offloading subtasks to $e^{h+1,}_{j}$.
We define $e^{h+1}_{j^{*}}$ as the receiver in $L^{h}_{i}$ that minimizes the value of $\Delta_{i,j}^{h,t}$, offloader $e^{h}_{i}$ updates the offloading strategy $p_{i,\cdot}^{h,t+1}$ by the following formula (Algorithm \ref{DTO-O}, Line 5):
\begin{equation}
	\left\{
	\begin{aligned}
		& p_{i,j}^{h,t+1} = (1-\tau_{p}) p_{i,j}^{h,t},     e^{h+1}_{j} \in L^{h}_{i}, j \neq j^{*}\\
		& p_{i,j^{*}}^{h,t+1} = p_{i,j^{*}}^{h,t} + \tau_{p} \sum\nolimits_{j\neq j^{*}}p_{i,j}^{h,t},    e^{h+1}_{j} \in L^{h}_{i}
	\end{aligned}
    \label{newp}
	\right.
\end{equation}
where $\tau_{p} \in(0,1]$ is the step size, and it determines the convergence speed of the offloading strategy.
After obtaining the new offloading strategy, $e^{h}_{i}$ gets its task arrival rate $\phi^{h,t}_{i}$  from Algorithm \ref{DTO-R} if $e^{h}_{i} \in S$.
Then it calculates the required computing resources in the next round based on the new offloading strategy and sends RUR messages to its receivers (Algorithm \ref{DTO-O}, Lines 8-9).

\subsection{Joint Optimization of Offloading Strategy and Confidence Threshold}\label{section:DTO-EE}
The DTO-EE algorithm is developed to jointly optimize the task offloading strategy and the confidence threshold during the configuration update phase.
At the beginning of the DTO-EE algorithm, each offloader node $e^{h}_{i} \in D\cup S/S^{H}$ takes the uniform offloading probability as the initial value, $\ie$,  $p^{h,0}_{i,j} = 1/|L^{h}_{i}|, \forall e^{h+1}_{j}\in L^{h}_{i}$.
Then it sends the initial RUR messages containing the confidence thresholds $C$ and the requested computing resources $\xi^{h,0}_{i,j} = p^{h,0}_{i,j}\phi^{h}_{i}I_{h}\alpha_{h+1}$ to each receiver node $e^{h+1}_{j} \in L^{h}_{i}$ (Algorithm \ref{DTO-EE}, Line 1).

\begin{algorithm}[!t] 
	\caption{Joint Optimization Algorithm (DTO-EE)} \label{DTO-EE} 
	\begin{algorithmic}[1]
        \Statex \textbf{Input:} The total communication round $n$, the update frequency $m$ 
        \Statex \textbf{Output:} The task offloading strategy $P$, the confidence thresholds $C$
        \State Each offloader initializes its task offloading strategy and sends initialized RUR messages to its receivers
        \For {Communication round $t=0,1,...,n$}
        \State All receivers and offloaders concurrently execute the Algorithm \ref{DTO-R} and Algorithm \ref{DTO-O}, respectively
        \State $h=(t/m)\% H$
        \If {$t\% m=0$ and $E_{h}=1$}
        \State Each offloader $e^{h}_{i}$ calculates $\Delta D^{h,t}_{i}$ by Eq. \eqref{DeltaD}
        \State Each offloader $e^{h}_{i}$ shares $\Delta D^{h,t}_{i}$ with each other and calculates $\Delta U$ by Eq. \eqref{DeltaU} 
        \State Select new confidence threshold that minimizes $\Delta U$
        \EndIf
        
        \EndFor
	\end{algorithmic} 
\end{algorithm}

In each subsequent communication round, to reduce the average response delay, all receivers concurrently execute the DTO-R algorithm, while all offloaders concurrently execute the DTO-O algorithm to update the task offloading strategy (Algorithm \ref{DTO-EE}, Line 3).
Besides, the confidence threshold is dynamically adjusted at the interval of several communication rounds to balance the response delay and inference accuracy. 
Given the update frequency $m$ and the step size $\tau_{c}$ for the confidence threshold, nodes cyclically update the confidence threshold for each branch.
Specifically, in every $m$ rounds of communication (\ie, when $t\%m=0$), if sub-model $M_{h}$ ($h=(t/m)\%H$) has an early exit branch, 
each node $e^{h}_{i}\in S^{h}$ calculates the impact of adjusting the confidence threshold $c_{h}$ by one step $\tau_{c}$ on the response delay $\Delta D^{h,t}_{i}$ (Algorithm \ref{DTO-EE}, Line 6).
Then, each node $e^{h}_{i} \in S^{h}$ shares $\Delta D^{h,t}_{i}$ with each other, and calculates the overall impact $\Delta U$ of increasing or decreasing the confidence threshold $c_{h}$ by one step $\tau_{c}$ on the system (Algorithm \ref{DTO-EE}, Line 7).
If the increase or decrease of the confidence threshold by one step $\tau_{c}$ can make $\Delta U \textless 0$, all node $e^{h}_{i}\in S^{h}$ make the same adjustment to the confidence threshold, and the adjusted threshold is sent to other nodes through RUR/RUS messages.
Otherwise, the nodes will not change the confidence threshold.
\subsection{Convergence Analysis}\label{sec:analysis}
Given the determined early exit confidence thresholds, we will analyze that our task offloading strategy can converge to an optimum.
We define $P^{t} = \{p^{h,t}_{i}|e^{h}_{i}\in S\cup D \}$ as the task offloading strategies of all edge nodes in the $t$-th round.
We use $\Gamma$ to represent the updating of the task offloading strategy in each round, \ie,
\vspace{-0.1cm}
\begin{equation}
   \begin{aligned}
       P^{t+1} = \Gamma(P^{t})
   \end{aligned}
   \vspace{-0.1cm}
\end{equation}
if $P^{t+1} = P^{t}$, $P^{t}$ is the fixed point of the updating, which is denoted as $P^{*}$. 
In the $t$-th round, the objective function $R(P)$ can be regarded as a continuous function of $p^{h,t}_{i,j}$.

\begin{lemma}
	If $P^{t} \neq P^{*}$, the updating $\Gamma$ gives a gradient descent direction of $R(P^{t})$ at point $P^{t}$, $\ie$,
 \vspace{-0.1cm}
\begin{equation}
   \begin{aligned}
       \langle\nabla R(P^{t}), \Gamma(P^{t}) - P^{t}\rangle \textless 0
   \end{aligned}
   \vspace{-0.1cm}
\end{equation}    
	where $\nabla$ is the gradient operator, and $\langle a, b \rangle$ represents the inner product of vectors a and b.
\end{lemma}
\begin{proof}
By Eqs. \eqref{partial_derivative} and \eqref{Omega}, we compute the partial derivative of $R(P^{t})$ \wrt $p_{i,j}^{h,t}$ as: 
\vspace{-0.1cm}
\begin{equation}
   \begin{aligned}
       \frac{\partial R(P^{t})}{\partial p_{i,j}^{h,t}} = \frac{\phi^{h}_{i} I_{h}}{\Phi}\Delta_{i,j}^{h,t}
   \end{aligned}
   \vspace{-0.1cm}
\end{equation}
Therefore,
\vspace{-0.1cm}
\begin{equation}
    \begin{aligned}
        &\langle\nabla R(P^{t}), \Gamma(P^{t}) - P^{t}\rangle \\
        &= \frac{1}{\Phi}\sum_{p^{h,t}_{i,j}\in P^{t}} \phi^{h}_{i} I_{h} \Delta_{i,j}^{h,t} (p_{i,j}^{h,t+1}-p_{i,j}^{h,t}) \\
        &= \frac{1}{\Phi}\sum_{e^{h}_{i}\in S\cup D} \phi^{h}_{i} I_{h} \sum_{e^{h+1}_{j}\in L^{h}_{i}}\Delta_{i,j}^{h,t}(p_{i,j}^{h,t+1}-p_{i,j}^{h,t})
    \end{aligned}
    \vspace{-0.1cm}
\end{equation}

According to Section \ref{section:DTO-O}, for offloader $e^{h}_{i}$, $e^{h+1}_{j^*}$ is the receiver that minimizes $\Delta^{h,t}_{i,j}$, so $\Delta_{i,j^*}^{h,t} \leq \Delta_{i,j}^{h,t}$ for any $e^{h+1}_{j}\in L^{h}_{i}$, and then,
\vspace{-0.1cm}
\begin{equation}
    \begin{aligned}
        &\sum_{e^{h+1}_{j} \in L^{h}_{i}} \!\!\! \Delta_{i,j}^{h,t}p_{i,j}^{h,t+1} \!\!= 
        \Delta^{h,t}_{i,j^*}p_{i,j*}^{h,t+1}\!\! +\!\!\!\! \sum_{\substack{e^{h+1}_{j} \in L^{h}_{i}\\
        j\neq j^{*}}} \!\! \Delta_{i,j}^{h,t}p_{i,j}^{h,t+1}\\
        &=\Delta^{h,t}_{i,j^*}(p_{i,j*}^{h,t} \!\! + \tau_{p} \!\! \!\sum_{\substack{e^{h+1}_{j} \in L^{h}_{i}\\j\neq j^{*}}} \!\!p_{i,j}^{h,t})\!\!+\!\! \sum_{\substack{e^{h+1}_{j} \in L^{h}_{i}\\j\neq j^{*}}} \!\! \Delta_{i,j}^{h,t}(1-\tau_{p})p_{i,j}^{h,t}\\
        &=\!\!\!\!\sum_{e^{h+1}_{j} \in L^{h}_{i}} \!\!\! \Delta_{i,j}^{h,t}p_{i,j}^{h,t} \!\!+\! \tau_{p} (\Delta_{i,j^*}^{h,t} \!\!\sum_{\substack{e^{h+1}_{j} \in L^{h}_{i}\\j\neq j^{*}}} \!\!\!\!\! p_{i,j}^{h,t} - \!\!\!\! \sum_{\substack{e^{h+1}_{j} \in L^{h}_{i}\\j\neq j^{*}}} \!\!\!\! \Delta_{i,j}^{h,t}p_{i,j}^{h,t})\\
        &\leq \sum_{e^{h+1}_{j} \in L^{h}_{i}}\Delta_{i,j}^{h,t}p_{i,j}^{h,t}  
    \end{aligned}
    \vspace{-0.1cm}
\end{equation}
Since $R(P^{t})$ is strictly convex and $P^{t} \neq P^{*}$, there is at least one node $e^{h+1}_{j}\in L^{h}_{i}$ satisfying $\Delta^{h,t}_{i,j} \textgreater \Delta^{h,t}_{i,j*}$. 
Therefore, $\sum_{e^{h+1}_{j} \in L^{h}_{i}}\Delta_{i,j}^{h,t}p_{i,j}^{h,t+1} \textless \sum_{e^{h+1}_{j} \in L^{h}_{i}}\Delta_{i,j}^{h,t}p_{i,j}^{h,t}$, and $\langle\nabla R(P), \Gamma(P^{t}) - P^{t}\rangle \textless 0$.
\vspace{-0.1cm}
\end{proof}
Since $R(P^{t})$ is differentiable on $\tau_{p}$, if $P^{t} \neq P^{*}$, there exit $\tau_{p} \in (0,1]$ such that $R(P^{t+1}) \textless R(P^{t})$.
Therefore, $R(P^{t})$ is a monotonically decreasing sequence for $t$, and $R(P^{t}) \textgreater 0$, then $R(P^{t})$ will finally converge.
Due to the space limitation, the detailed proof is omitted.

%% file: content/experimentation.tex
\subsection{Experimental Settings}\label{Simulation_Settings}

\textbf{System Deployment.}
The simulation experiments are conducted on an AMAX deep learning workstation equipped with an Intel(R) Xeon(R) Gold 5218R CPU, 8 NVIDIA GeForce RTX 3090 GPUs and 256 GB RAM.
The system timeline is divided into 5s time slots, incorporating a 100ms configuration update phase, and the local communication delay is set at 2ms\cite{ma2023fully}.
We employ 50 EDs to generate tasks following a Poisson process and deploy each sub-model across 4-6 ESs (distribution skewed towards fewer ESs for later sub-models due to the early exit mechanism), and each offloader node is assigned 2-4 receiver nodes for task offloading.
The computing capacity of each ES is obtained from 3 mainstream devices, \ie, Jetson TX2, Jetson NX, and Jetson AGX. 
Each device can be configured to work with different modes, specifying the number of working CPUs and the frequency of CPUs/GPUs.
Devices operating in different modes exhibit different capabilities. 
Concretely, the fastest mode (\ie, mode 0 of AGX) achieves inference speeds approximately 5× faster than the slowest mode (\ie, mode 1 of TX2). 
We employ the recorded computing capacity of these devices in the simulation framework to simulate the heterogeneous system.
Besides, the bandwidth from ED to ES is set to vary from 1MB/s to 10MB/s\cite{teng2022game}, while the bandwidth between ESs is set to vary from 10MB/s to 20MB/s.

\begin{table}[t]

	\caption{The Parameters Details of Sub-models.}
    \vspace{2mm}
	\label{model_partition}
	\centering
        \setlength{\tabcolsep}{1mm}
        \scalebox{0.9}{
	\begin{tabular}{c|c|ccccc}
		\hline
		\multirow{2}{*}{Model} & \multirow{2}{*}{Parameter} & \multicolumn{5}{c}{Sub-model} \\
		\cline{3-7}
		&&\textbf{$h_{1}$}& \textbf{$h_{2}$}& \textbf{$h_{3}$}& \textbf{$h_{4}$}& \textbf{$h_{5}$}\\
		\hline		
		\multirow{3}{*}{ResNet101} & $\alpha$(GFLOPs) & 2.21 & 1.97 & 1.97 & 1.68 & - \\
        & $\beta$(MB) & 0.14 & 0.77 & 0.77 & 0.77 & -\\
        & $\mathcal{A}$ & - & 0.470 & 0.582 & 0.681 & -\\
		\hline
		\multirow{3}{*}{BERT} & $\alpha$(GFLOPs) & 6.44 & 8.05 & 8.08 & 8.08 & 8.08 \\
		& $\beta$(MB) & 0.01 & 0.56 & 0.56 & 0.56 & 0.56\\
        & $\mathcal{A}$ & - & 0.552 & 0.568 & 0.572 & 0.582\\
		
		\hline
		
	\end{tabular}
       
 }
\end{table}

\textbf{Datasets and Models.}
We choose two classical datasets and two DNN models for performance evaluation.
1) \textit{ImageNet}\cite{russakovsky2015ImageNet} is a dataset for image recognition that consists of 1,281,167 training images, 50,000 validation images and 100,000 test images from 1,000 categories.
We adopt a famous model ResNet101 with a size of 171MB\cite{he2016deep} for image recognition. 
2) \textit{Tnews} \cite{xu2020clue} is a dataset for news text classification that consists of 53,360 training news, 10,000 validation news and 10,000 test news from 15 categories.
We adopt a famous large model Bert-large (denoted as Bert) with a size of 1.3G \cite{devlin2018bert} for the text classification. 

These models are pre-trained and divided into multiple sub-models, and exit branches are set on some sub-models.
Concretely, We divide ResNet101 into 4 sub-models and add exit branches on the second and third sub-models. Besides, we divide Bert into 5 sub-models and add exit branches on the second, third, and fourth sub-models.
Table 2 summarizes the computing resources required and the input size for inference in each sub-model, as well as the inference accuracy of all branches and the complete model.

\textbf{Baselines.}
We compare DTO-EE with four baselines.
\begin{itemize}
\item \textbf{Computing-First (CF)}: A heuristic task offloading algorithm, where each offloader offloads tasks in proportion to the computing capacity of receivers.

\item \textbf{Bandwidth-First (BF)}: A heuristic task offloading algorithm. where tasks are offloaded proportionally to the bandwidth resources of receivers.

\item \textbf{NGTO}\cite{tiwary2018response}: A non-cooperative game based task offloading algorithm where each offloaded updates its offloading strategy circularly. Each offloader performs a selfish decision based on the strategies of other offloaders,  passing these updates to subsequent offloaders. This process continues until a Nash equilibrium is achieved.

\item \textbf{GA}\cite{peng2024collaborative}: A genetic algorithm is executed by each ED, which continuously collects the state information of all nodes and searches out the path with the shortest delay for task offloading.
\end{itemize}

We adaptively adjust confidence thresholds across all baselines to balance delay and accuracy, which is set to the same update frequency and step size as DTO-EE.

\textbf{Metrics.}
We adopt the following metrics to evaluate the performance of DTO-EE and the baselines.
1) \textbf{Average response time} reflects the response rate of the collaborative inference framework, and is measured by the average delay of all samples from its arrival at the ED to the completion of inference.
2) \textbf{Inference accuracy} is calculated by the proportion of the tasks correctly predicted by the models to all arriving tasks, and is mainly affected by the confidence thresholds.

\subsection{Overall Performance}

\begin{figure}[!t]
	\centering
	\subfigure[Average response delay]
	{
		\includegraphics[width=0.46\linewidth,height=3.5cm]{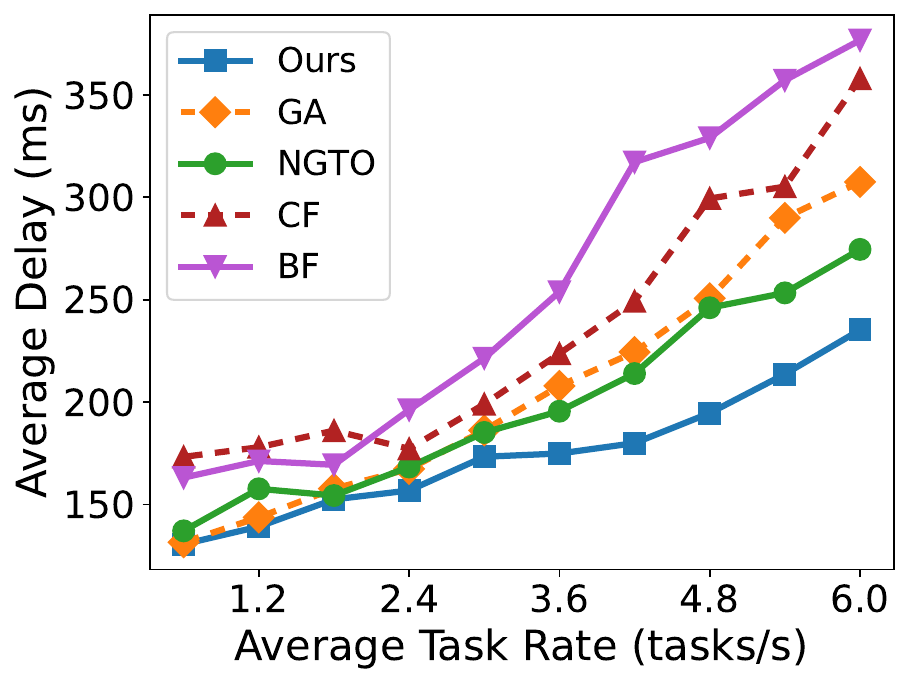}
		\label{fig:arrival rate1}
	}
	\subfigure[Inference accuracy]
	{
		\includegraphics[width=0.46\linewidth,height=3.5cm]{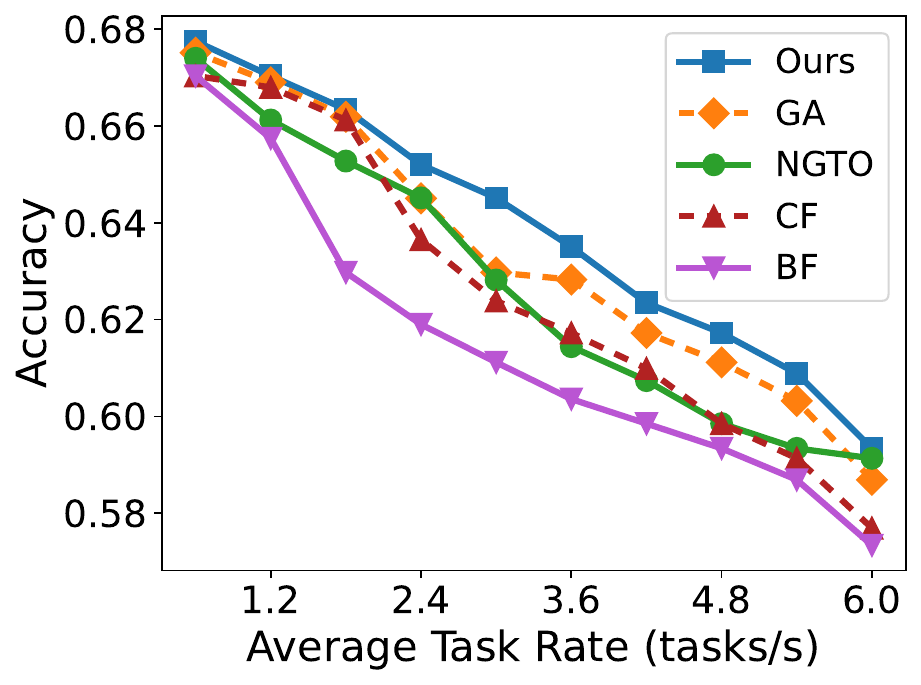}
		\label{fig:arrival rate2}
	}
 
	\caption{Inference performance of algorithms given different task arrival rates for ResNet101 on ImageNet.}
	\vspace{-0.2cm}
	\label{fig:arrival rate ImageNet}
\end{figure}

\begin{figure}[!t]
	\centering
	\subfigure[Average response delay]
	{
		\includegraphics[width=0.46\linewidth,height=3.5cm]{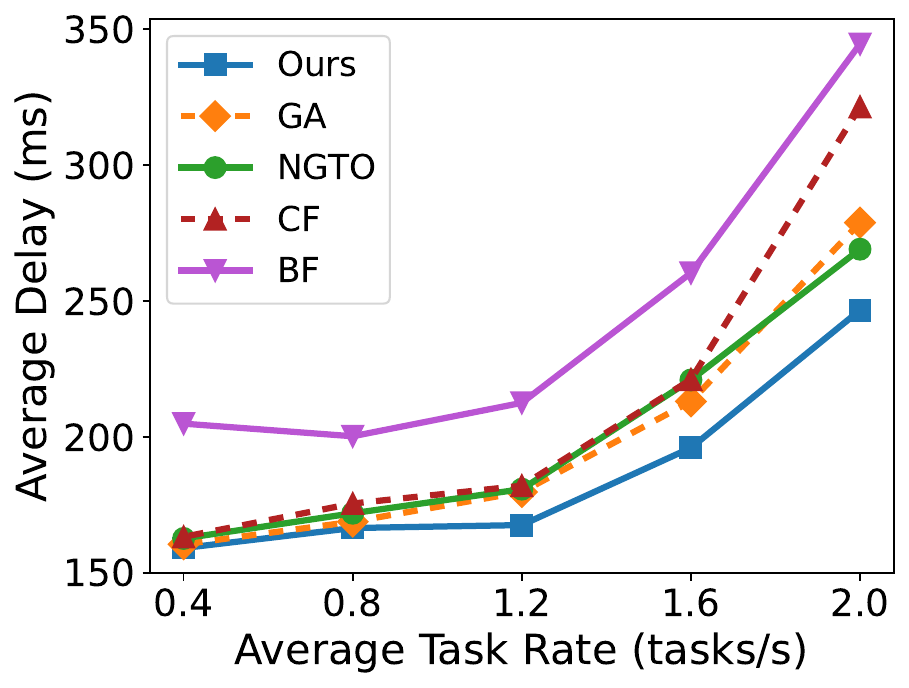}
		\label{fig:arrival rate3}
	}
	\subfigure[Inference accuracy]
	{
		\includegraphics[width=0.46\linewidth,height=3.5cm]{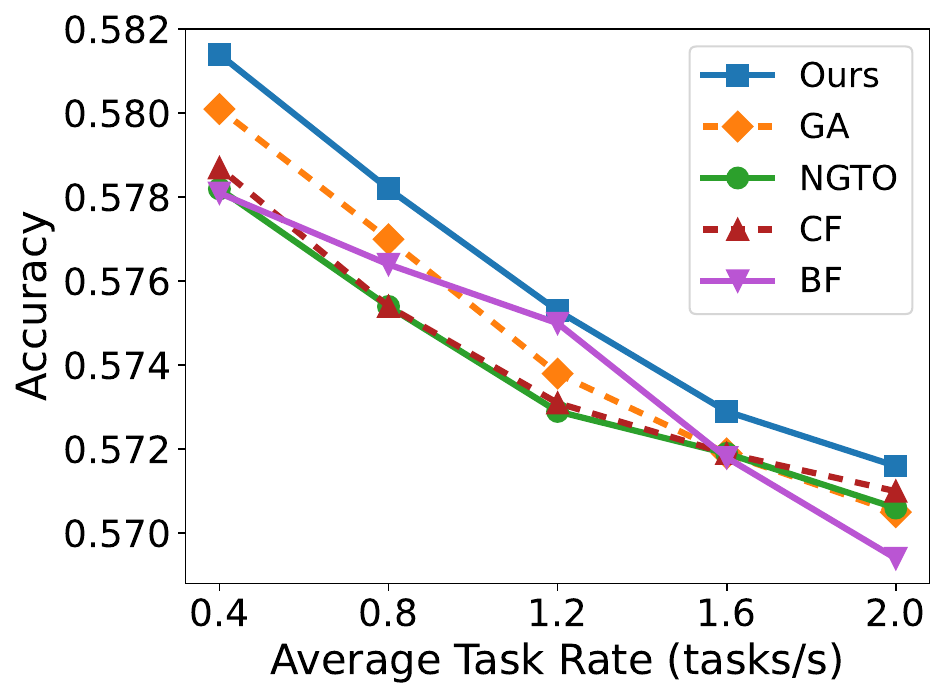}
		\label{fig:arrival rate4}
	}
 
	\caption{Inference performance of algorithms given different task arrival rates for Bert on Tnews.}
	\vspace{-0.4cm}
	\label{fig:arrival rate tnews}
        
\end{figure}

 

 

\begin{figure}[t]
	\centering
	\subfigure[Average response delay]
	{
		\includegraphics[width=0.46\linewidth,height=3.5cm]{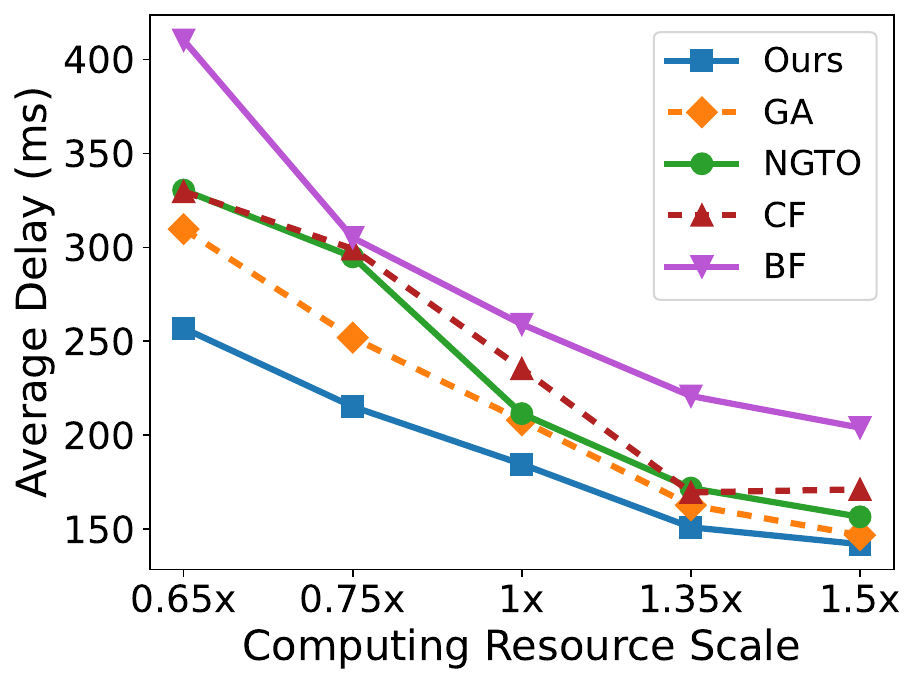}
		\label{fig:cmp1}
	}
	\subfigure[Inference accuracy]
	{
		\includegraphics[width=0.46\linewidth,height=3.5cm]{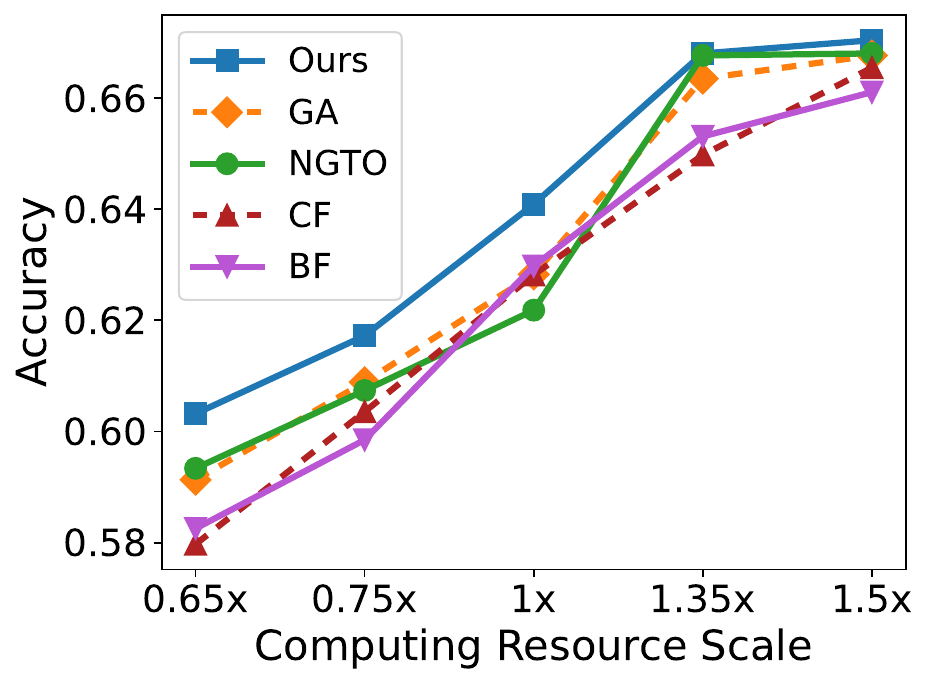}
		\label{fig:cmp2}
	}
    \caption{Inference performance with varying average computing resource for RestNet101 on ImageNet.}
    \label{fig:cmp ImageNet}
    \vspace{-2mm}
\end{figure}

\begin{figure}[t]
	\centering
        \subfigure[Average response delay]
	{
		\includegraphics[width=0.46\linewidth,height=3.5cm]{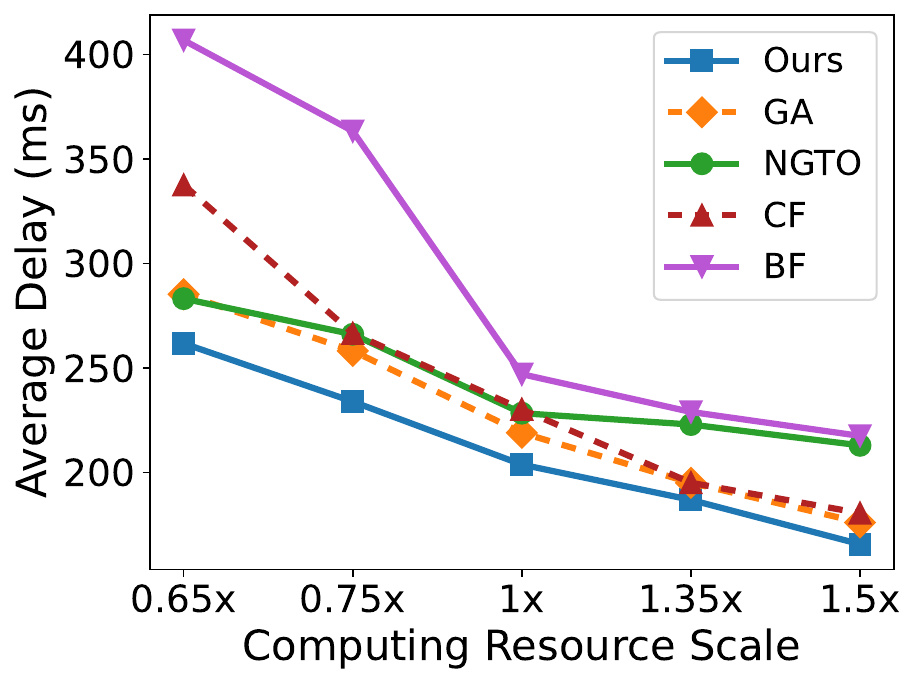}
		\label{fig:cmp3}
	}
	\subfigure[Inference accuracy]
	{
		\includegraphics[width=0.46\linewidth,height=3.5cm]{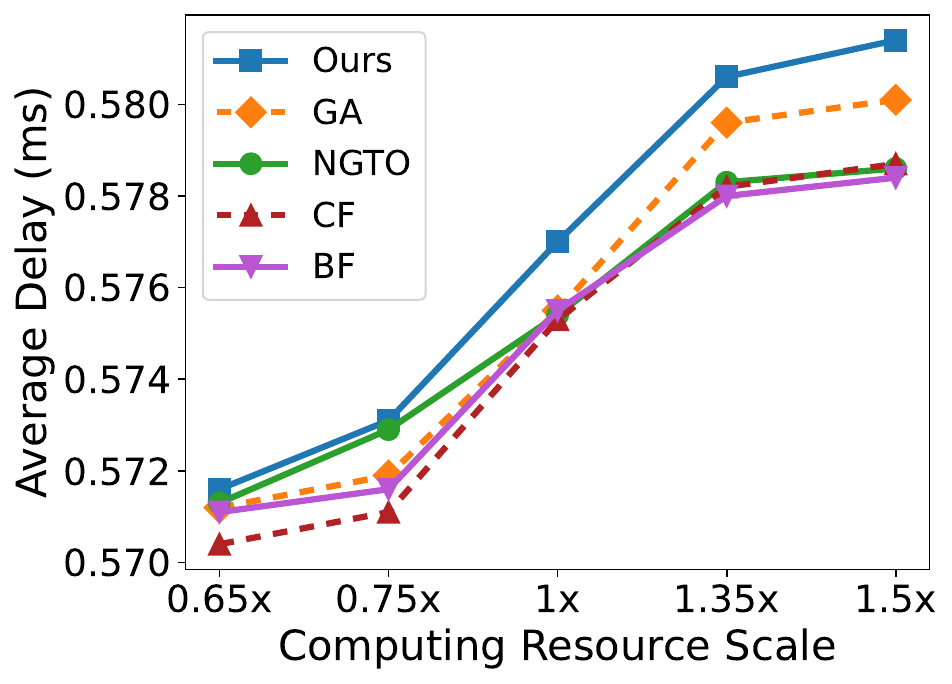}
		\label{fig:cmp4}
	}
        \vspace{-2mm}
	\caption{Inference performance with varying average computing resource for Bert on Tnews.}
	\label{fig:cmp tnews}
        \vspace{-4mm}
\end{figure}

We conduct a set of experiments with different arrival rates of tasks to evaluate the inference performance of DTO-EE and the baselines.
The average response delay and inference accuracy of five approaches on two datasets are presented in Fig. \ref{fig:arrival rate ImageNet} and Fig. \ref{fig:arrival rate tnews}, respectively.
We observe that as the task arrival rate increases, the average response delay of all approaches rises, accompanied by a varying degree of accuracy degradation.
Nevertheless, DTO-EE achieves the lowest average response delay and the highest inference accuracy, consistently outperforming other approaches on both datasets, and its advantage is more obvious with a heavier computing load.
Within the heterogeneous system, CF and BF focus solely on computing resources and transmission rates, respectively, which may lead to excessive transmission delay or uneven computing load.
NGTO only performs a local optimal offloading strategy for the current subtask without considering the impact of the strategy on the offloading process of subsequent subtasks, which may lead to poor overall performance.
GA constantly collects the resource information from all nodes to search for a path with the shortest delay for task offloading.
However, each ED selfishly optimizes its task offloading strategy, which leads to an unbalanced computing load and an increase in overall delay.
Concretely, by Fig. \ref{fig:arrival rate ImageNet}, when the average task arrival rate is 4.8 tasks/s for ResNet101 on ImageNet, DTO-EE achieves the average response delay of 195ms and 61.7\% inference accuracy, while GA, NGTO, CF, and BF achieve an average response delay of 250ms, 246ms, 299ms, 329ms, and their inference accuracy is 61.1\%, 59.8\%, 59.8\%, and 59.3\%, respectively.
Besides, when the average task arrival rate is 2 tasks/s for Bert on Tnews, Fig. \ref{fig:arrival rate tnews} shows that DTO-EE reduces the average response delay by about 11.6\%, 8.6\%, 23.3\%, and 28.4\%, and improves the inference accuracy by about  1.9\textperthousand, 1.8\textperthousand, 1.1\textperthousand, and 3.9\textperthousand, compared to GA, NGTO, CF and BF, respectively.




Secondly, we also conduct a set of experiments on these approaches with different computing resources.
We adjust the computing mode of ESs to scale the average computing resource of the system, and the results are presented in Fig. \ref{fig:cmp ImageNet} and Fig. \ref{fig:cmp tnews}.
We observe that DTO-EE still outperforms other approaches with different computing resources.
For instance, by Fig. \ref{fig:cmp ImageNet}, with the resource-constrained case where the computing resource is shrunk to 0.65$\times$, DTO-EE achieves an average response delay of 257ms and 60.3\%  inference accuracy for ResNet101 on ImageNet, while GA, NGTO, CF, and BF achieve the average response delay of 310ms, 330ms, 329ms, and 410ms, and reach the inference accuracy of 59.1\%, 59.3\%, 57.9\%, and 58.3\%, respectively. 
Besides, as shown in Fig. \ref{fig:cmp tnews}, with the resource-rich case where the computing resource is scaled up to 1.5$\times$, the average response delay of DTO-EE for Bert on Tnews can be reduced by 6.0\%, 22.2\%, 8.2\%, and 23.8\%, and the inference accuracy can be improved by 2.2\textperthousand, 4.8\textperthousand, 4.7\textperthousand, and 5.2\textperthousand, compared to GA, NGTO, CF, and BF, respectively.
These results demonstrate the superiority of DTO-EE in addressing system heterogeneity.

\subsection{Effect of Dynamic Environment}\label{sec: Dynamic Environment}

\begin{figure}[t]
	\centering
	\subfigure[Average response delay]
	{
		\includegraphics[width=0.46\linewidth,height=3.5cm]{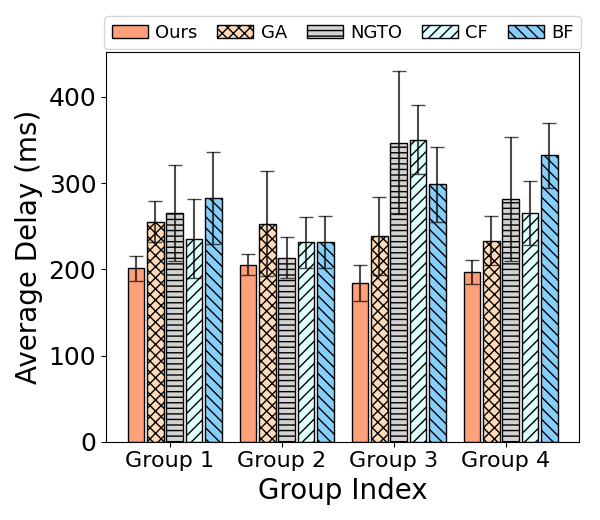}
		\label{fig:dynamic1}
	}
	\subfigure[Inference accuracy]
	{
		\includegraphics[width=0.46\linewidth,height=3.5cm]{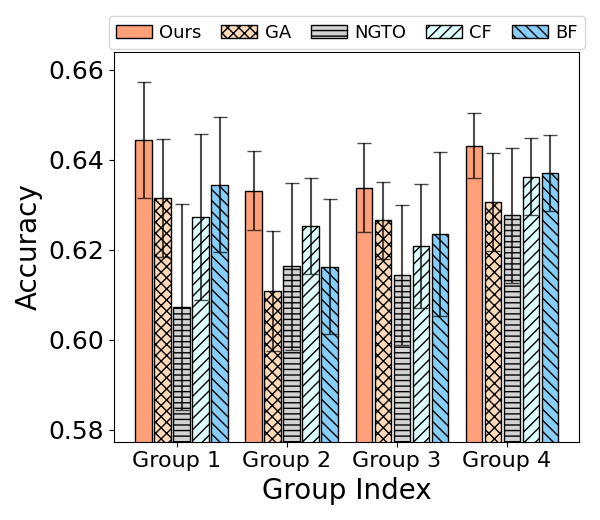}
		\label{fig:dynamic2}
	}
	\caption{Inference performance in the dynamic environment for ResNet101 on ImageNet.}
        \vspace{-2mm}
	\label{fig:dynamic ImageNet}
\end{figure}

\begin{figure}[t]
	\centering
	\subfigure[Average response delay]
	{
		\includegraphics[width=0.46\linewidth,height=3.5cm]{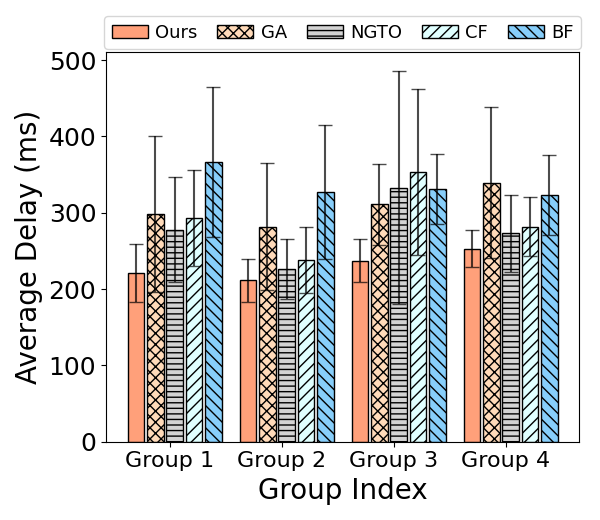}
		\label{fig:dynamic3}
	}
	\subfigure[Inference accuracy]
	{
		\includegraphics[width=0.46\linewidth,height=3.5cm]{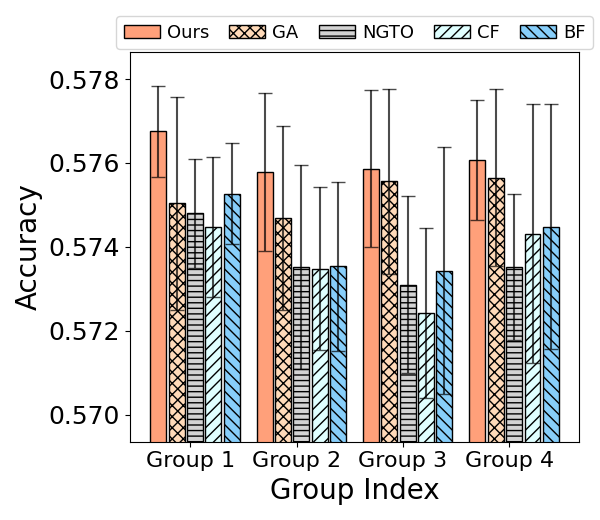}
		\label{fig:dynamic4}
	}
        \vspace{-2mm}
	\caption{Inference performance in the dynamic environment for Bert on Tnews.}
        \vspace{-4mm}
	\label{fig:dynamic Tnews}
\end{figure}

To demonstrate the robustness of DTO-EE, we evaluate the performance of DTO-EE and baselines in a dynamic environment.
At the end of each time slot, we randomly change the task arrival rate of EDs and adjust the computation mode of ESs to simulate the dynamic environment. 
We record the average response delay and accuracy for each time slot, and integrate the data for every 10 time slots into one group.
As shown in Fig. \ref{fig:dynamic ImageNet} and Fig. \ref{fig:dynamic Tnews}, DTO-EE outperforms other approaches in terms of delay and accuracy. 
NGTO requires the task offloaders to optimize their task offloading strategies in a cyclic manner up to Nash equilibrium, which leads to a longer decision time and poor performance in dynamic environments. 
CF and BF can quickly adjust the offloading strategy but cannot balance the computing load of heterogeneous nodes. 
GA relies on the state information of all ESs, which needs to be forwarded multiple times to reach the ED and may be outdated due to interference from other EDs. 
This outdated information may cause the tasks of multiple EDs to be concentrated on a few paths without effectively balancing the computing load.
For example, by Fig. \ref{fig:dynamic ImageNet}, DTO-EE achieves the average response delay of 184ms and 63.4\% accuracy in Group 3 on ImageNet, while GA, NGTO, CF, and BF achieve the average response delay of 238ms, 347ms, 350ms, and 298ms, and reach the inference accuracy of 62.7\%, 61.4\%, 62.1\%, and 62.4\%, respectively. 
Besides, as shown in Fig. \ref{fig:dynamic Tnews}, DTO-EE achieves the lowest and most stable response delay while maintaining the highest inference accuracy in each group for Bert on Tnews, the average standard deviation of the delay of DTO-EE in the four groups is 29ms, while the average standard deviation of the delay of GA, NGTO, CF, and BF is 84ms, 78ms, 63ms, and 71ms, respectively.
These results demonstrate that DTO-EE is effective and efficient in addressing the challenge of the dynamic environment.

\subsection{Effect of Dynamic Threshold}

\begin{figure}[t]
	\centering
	\subfigure[Average response delay]
	{
		\includegraphics[width=0.46\linewidth,height=3.5cm]{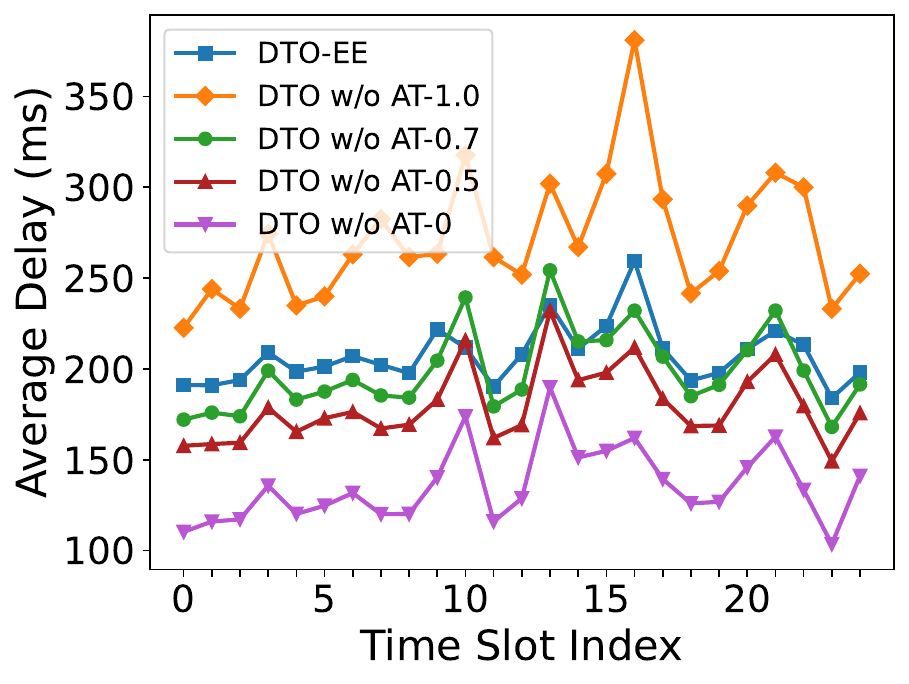}
		\label{fig:aba delay}
	}
	\subfigure[Inference accuracy]
	{
		\includegraphics[width=0.46\linewidth,height=3.5cm]{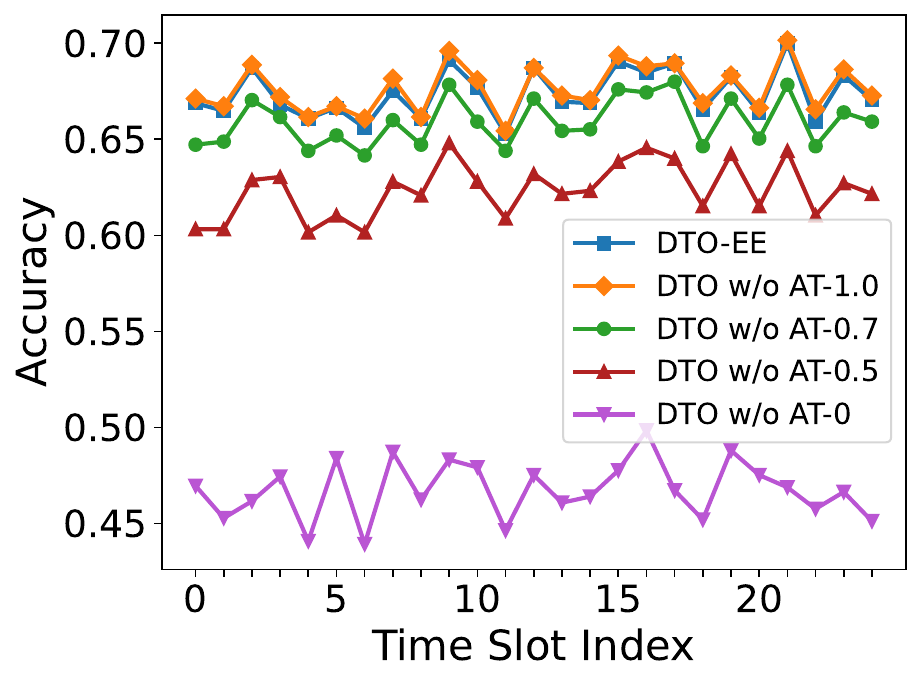}
		\label{fig:aba acc}
	}
        \vspace{-1mm}
	\caption{Effect of dynamic thresholds on ImageNet.}
	\label{fig:aba ImageNet}
        \vspace{-4mm}
\end{figure}

Dynamically adjusting confidence thresholds is developed to balance the response delay and inference accuracy.
Herein, we conduct multiple experiments within a dynamic environment to verify the effectiveness of dynamically adjusting confidence thresholds. 
We adopt the DTO-EE without adjusting thresholds (DTO w/o AT) as the baselines. 
For instance, in DTO w/o AT-0.7, confidence thresholds of all exit branches are fixed at 0.7, and each node only optimizes its task offloading strategy.
We initially deploy each sub-model to the same number of multiple ESs with the same computing power, and the transmission rate between nodes is the same. 
We test 25 time slots in the dynamic environment described in Sec. \ref{sec: Dynamic Environment}, recording the average response delay and inference accuracy of tasks in each time slot.
By Fig. \ref{fig:aba ImageNet}, compared to DTO-EE w/o AT-1.0 without early exiting, DTO-EE reduces the average response delay by 23.5\% 
while achieving almost the same inference accuracy.
Compared to DTO-EE w/o AT-0.7, DTO-EE improves the inference accuracy by 2.2\% at the expense of only 4.3\% average response delay. 
These results reflect the positive role of dynamically adjusting confidence thresholds.

%% file: content/related.tex
Offloading computation-intensive tasks from end devices to edge servers can effectively reduce task response delay and energy consumption.
Chen \etal \cite{chen2015efficient} explore the multi-user computation offloading problem for mobile-edge cloud computing in a multi-channel wireless interference environment, and adopt a game theoretic approach for achieving efficient task offloading in a distributed manner.
Similarly, Chen \etal \cite{chen2018task} focus on task offloading in a software-defined ultra-dense edge network, with the objective of minimizing the delay while saving the battery life of end devices. 
Tang \etal \cite{multiserver} propose a distributed algorithm based on deep reinforcement learning (DRL) to optimize the offloading strategies of end devices without knowing the computing load of edge servers.
Xiao \etal \cite{xiao2018distributed} present an offloading strategy where various fog nodes with different computing and energy resources  work cooperatively, they use the subgradient method with dual decomposition to reduce the service response delay and improve the efficiency of power usage.
Besides, Xiao \etal \cite{xiao2022edge} propose an offloading algorithm based on deep reinforcement learning to offload several subtasks with dependent relations in the dynamic environment.

The collaborative inference of multiple nodes can be realized by model partitioning, which can effectively reduce the computing load of a single node.
Gao \etal \cite{gao2023Task} divide the tasks between end devices and edge servers, and design a distributed algorithm based on aggregate game theory to optimize model partitioning and offloading strategies.
Similarly, Xu \etal \cite{xu2023cnn} propose a dynamic offloading strategy based on game theory combined with convolutional neural network (CNN) partition, which enables tasks to be inferred in parallel on multiple edge nodes to reduce delay.
Mohammed \etal \cite{2partition2} propose a distributed solution based on matching theory for joint partitioning and offloading of DNN inference tasks to curtail computing delay and enhance resource utilization.

The early exit mechanism can be integrated into collaborative inference to further decrease task response delay.
Li \etal \cite{Edgent} propose collaborative inference between end devices and edge servers, by predicting the delay of each layer of the network, the tasks exit early at the appropriate place to maximize the accuracy of the model inference.
Liao \etal \cite{liu2023resource} investigate the joint optimization of early exit and batching in task offloading to improve system throughput.
Ju \etal \cite{eDeepSave} focus on accelerating DNN inference with an early exit mechanism to maintain performance during mobile network transitions.
Liu \etal \cite{exit1} propose an adaptive edge inference framework based on model partitioning and multi-exit model, and balance the number of completed tasks and average inference accuracy by optimizing exit and partition points for each task.

%% file: content/conclusion.tex
In this paper, we introduce a cooperative inference framework based on model partitioning and early exit mechanism.
To address the challenges of the heterogeneous system and dynamic environment, we theoretically analyze the coupled relationship between task offloading strategy and confidence thresholds of different exit branches.
Based on the coupled relationship and convex optimization, we develop a distributed algorithm, called DTO-EE, which enables each node to jointly optimize its offloading strategy and the confidence threshold to achieve a promising trade-off between response delay and inference accuracy.
The experimental results show that DTO-EE reduces the response delay by 21\%-41\% and improves the inference accuracy by 1\%-4\%, compared to the baselines.